\documentclass[preprint]{elsarticle}
\usepackage{multicol}
\usepackage{tikz}
\usepackage{empheq}
\newcounter{casenum}

\usepackage[T1]{fontenc}
\usetikzlibrary{patterns,positioning}
\usetikzlibrary{shapes.geometric,calc,decorations.text}
\usepackage{caption}
\usepackage{subcaption}
\usepackage{todonotes}
\usepackage{bbold} 
\usepackage{float}
\usepackage{amsmath}
\usepackage{amsfonts,amssymb,amsmath}
\usepackage{xcolor}
\usepackage{url}
\usepackage{subcaption}
\usepackage{algorithm}
\usepackage[justification=centering]{caption}
\usetikzlibrary{patterns,positioning}
\usepackage{indentfirst}
\usepackage{caption}
\usepackage{amsthm}
\usepackage{amsmath}

\usetikzlibrary{shapes.misc}

\usetikzlibrary{shapes.geometric}
\usetikzlibrary{patterns,positioning}
\usepackage{amsmath, amssymb}
\usepackage{tikz}
\usetikzlibrary{shapes.geometric,calc,decorations.text}
\usetikzlibrary{patterns}

\usepackage{mathrsfs}
\usepackage{amssymb}
\usepackage{dsfont}
\usetikzlibrary{patterns}
\usetikzlibrary{shapes.geometric}
\def\mytransformation{
	\pgfmathsetmacro{\myX}{0.9\pgf@x}
	\pgfmathsetmacro{\myY}{0.0045*(\pgf@x)*(\pgf@x)+(\pgf@y)}
	\setlength{\pgf@x}{\myX pt}
	\setlength{\pgf@y}{\myY pt}
}

\makeatother
\usetikzlibrary{backgrounds}
\usetikzlibrary{shapes.geometric}
\usetikzlibrary{calc}
\usepackage{enumerate}
\usepackage{mathrsfs}
\usepackage{amssymb}
\usepackage{graphicx}%
\usepackage{multirow}%
\usepackage{amsmath,amssymb,amsfonts}%
\usepackage{mathrsfs}%
\usepackage{xcolor}%
\usepackage{textcomp}%
\usepackage{manyfoot}%
\usepackage{algpseudocode}%
\usepackage{setspace}
\usepackage{pdflscape}
\usepackage{graphicx}
\usepackage{pst-node}
\usepackage{booktabs}
\usepackage{tikz}
\usepackage{todonotes}
\usepackage[inline]{enumitem}

\renewenvironment{proof}{{\it Proof.}}{\qed \medskip}

\newtheorem{theorem}{Theorem}
\newtheorem{obs}{Observation}
\newtheorem{corollary}[theorem]{Corollary}
\newtheorem{lemma}[theorem]{Lemma}

\usepackage{mathtools}

\providecommand{\keywords}[1]{\textbf{\textit{Keywords—}} #1}

\begin{document}
	
	\title{A column generation algorithm for finding co-3-plexes in chordal graphs}
	
	\author[1,2]{Alexandre Dupont-Bouillard}
	
	\affiliation[1]{\organization={Université de Rennes -- Institut de recherche en informatique et systèmes aléatoires CNRS UMR 6074},  \adressline={263 Av. Général Leclerc}, \city={Rennes}, \postcode={35000}, \country={France}}
	
	\affiliation[2]{\organization={Ecole polytechnique, Institut Polytechnique de Paris LIX -- CNRS UMR 7161}, \adressline={263 Avenue général Leclerc},\city={Palaiseau}, \postcode={91120}, \country={France} }

	\journal{Discrete optimization}
	
	\begin{abstract}
		In this study, we tackle the problem of finding a maximum \emph{co-3-plex}, which is a subset of vertices of an input graph, inducing a subgraph of maximum degree 2. We focus on the class of chordal graphs.
		By observing that the graph induced by a co-3-plex in a chordal graph is a set of isolated triangles and induced paths, we reduce the problem of finding a maximum weight co-3-plex in a graph $G$ to that of finding a maximum stable set in an auxiliary graph $\mathcal{A}(G)$ of exponential size. 
		This reduction allows us to derive an exponential variable-sized linear programming formulation for the maximum weighted co-3-plex problem.
		We show that the pricing subproblem of this formulation reduces to solving a maximum vertex and edge weight induced path.
		Such a problem is solvable in polynomial time; therefore, this exhibits a polynomial time column generation algorithm solving the maximum co-3-plex problem on chordal graphs. 
		Moreover, this machinery exhibits a new application for the maximum vertex and edge weighted induced path problems.
		
		
	\end{abstract}
	
	\maketitle

	\section*{Definitions}
	All graphs in this article are simple and connected unless specified. 
	Given a graph $G=(V,E)$, we denote its \textit{complement} by $\overline{G}=(V,\overline E)$, where $\overline{E} = \{ e \in \binom{V}{2} : e \notin E \}$. 
	We respectively denote by $V(G)$ and $E(G)$ the vertex set and the edge set of $G$. Two vertices $u$ and $v$ are \emph{adjacent} if $uv \in E(G)$. 
	Given a subset of vertices $W \subseteq V$, let $E(W)$ denote the set of edges of $G$ with both endpoints in $W$ and $\delta (W)$ denote the set of all edges with exactly one endpoint in $W$.
	When $W$ is a singleton $\{w\}$, we will write $\delta (w)$. We say that the edges in $\delta (w)$ are \textit{incident} to $w$, and two edges sharing an extremity are said to be {\em adjacent}.
	For $F \subseteq E$, let $V(F)$ denote the set of vertices incident to any edge in $F$.
	Given $W\subseteq V$, the graph $G[W]=(W,E(W))$ is the \textit{subgraph of $G$ induced by $W$}. When $H$ is an induced subgraph of $G$, we say that $G$ {\em contains} $H$. 
	Given a vertex $u\in V(G)$, we denote  by $N_G(u) = \{w \in V(G) : uw \in E(G) \}$ its \textit{neighborhood}, and by $N_G[u] = N_G(u) \cup \{ u\}$ its \textit{closed neighborhood}, when clear from context we simply write $N(u)$ and $N[u]$. Two vertices $u$ and $v$ are {\em true twins} if $N[u] =N[v]$ and {\em false twins} if $N(u) = N(v)$.
	A vertex is \emph{simplicial} if its neighborhood induces a clique.
	
	A subset $K$ of vertices is a \emph{$k$-plex} (resp. \emph{co-$k$-plex}) if $G[K]$ has minimum degree $|K| - k$ (resp. maximum degree $k-1$). Note that $k$-plexes are complements of co-$k$-plexes. \emph{Cliques} are $1$-plexes and \emph{stable sets} are co-1-plexes. 
	We denote by $\mathcal{K}(G)$ the set of inclusion-wise maximal cliques of $G$.
	
	An \textit{induced path} (resp. \textit{hole}) is a graph induced by a sequence of vertices $(v_1,\dots,v_p)$ whose edge set is $\{ v_iv_{i+1}:  \ i = 1,\dots, p-1 \} $ (resp. $\{ v_iv_{i+1}:  \ i = 1,\dots, p-1 \} \cup \{v_1v_p \}$ ). 
	A subset of vertices induces a path (resp. hole) if its elements can be ordered into a sequence inducing a path (resp. hole).
	The {\em length} of a path or hole is its number of edges. The { \em parity} of a path or hole is the parity of its length.

	The \textit{contraction} of an edge $uv\in E(G)$ yields a new graph $G/uv$ built from $G$ by deleting $u$ and $v$, adding a new vertex $w$, and adding the edges $wz$ for all $z\in (N_G(u) \cup N_G(v)) \setminus \{ u,v\}$. 
	For $F \subseteq E$, we denote by $G/F$ the graph obtained from $G$ by contracting all edges in $F$. 
	
	A graph is $t$-chordal if its holes have maximum size $t$; 3-chordal graphs are called \emph{chordal}. 
	\medskip
	
	Let $\mathcal{P} = \{x:Ax \le b\}\subseteq \mathds{R}^n$ be a \emph{polyhedron}, that is, the intersection of finitely many half-spaces. The \emph{dimension} of $\mathcal{P}\subseteq \mathds{R}^n$, denoted by dim $\mathcal{P}$, is the maximum number of affinely independent points in $\mathcal{P}$ minus one. 
	If $a \in \mathds{R}^n\setminus \{0\}$, $\alpha \in \mathds{R}$, then the inequality $a^\top x \le \alpha$ is said to be \emph{valid} for $\mathcal{P}$ if $\mathcal{P} \subseteq \{ x \in \mathds{R}^n : a^\top x \le \alpha \}.$  
	The polyhedron $\mathcal{P} \subseteq \mathds{R}^n$ is \emph{ full-dimensional } when dim~$\mathcal{P}$ is equal to $n$.
	\emph{Extreme points} of a polyhedron are its faces of dimension zero.
	Each extreme point satisfies $\textrm{dim}~\mathcal{P}$ linearly independent valid inequalities for $\mathcal{P}$ with equality. 
	A polyhedron $\mathcal{P} \subseteq \mathds{R}^n$ is said to be \emph{integer} if all its non-empty faces contain an integer point.
	A \emph{polytope} is a bounded polyhedron, equivalently, it can be defined as the convex hull of a finite number of points $P$ and is denoted by $\textrm{conv} (P)$. 
	
	\medskip
	
	A 0/1 minimal squared submatrix of odd size having two ones per row and per column is called an \emph{odd hole} and has a determinant of two. 
	
	\medskip 
	
	Given $x$ variables on the vertices of a graph $G$, an important ingredient for our proof is the following theorem characterizing the perfectness of a graph by the hyperplane representation of the convex hull of its stable sets' incidence vectors, \emph{ that is } \emph{its stable set polytope}. Note that we will not use the historical definition of perfect graphs, which requires equality between the coloring number and clique number for each induced subgraph.
	
	\begin{theorem}[\cite{CHVATAL1975138}]\label{the:WPGT}
		A graph $G$ is perfect if and only if its stable set polytope is described as follows:
		\begin{align}
			x(K) \le 1 && \forall K \in \mathcal{K} \\
			x_u \ge 0 && \forall u \in V
		\end{align} 
	\end{theorem}
	
	Note that odd holes are not perfect, nor are their complement; these two types of induced subgraphs are the ones to be forbidden to describe perfect graphs, as stated in the strong perfect graphs theorem~\cite{perfectgraphtheorem}.
	
	\section*{Motivations }

	The problem of finding "communities" in a network arises in social network analysis, notably for the recognition of cohesive subgroups. The most natural graph notion one could think of to describe such a group is a clique; however, asking that all pairs of vertices are adjacent may be too restrictive. Therefore, several relaxed clique definitions have been proposed; see~\cite{PATTILLO20139} for a survey. One type of such relaxed cliques is called a $k$-plex, and finding a maximum weighted such subset of vertices is NP-hard~\cite{bala}. By complementing the input graph, the problem of finding a maximum weighted co-$k$-plex is also.   This problem (or rather its complement) has been widely studied from an experimental point of view~\cite{grasp,PATTILLO20139,bala,WANG201779,ya,Zhou_Hu_Xiao_Fu_2021,roberto2}. The case $k=2$ has also been studied from a polyhedral point of view in natural variable space~\cite{pol} and extended variable space~\cite{dupontbouillard2025extendedformulationsmaximumweighted}, permitting to describe the polytope of contraction perfect graphs~\cite{dupontbouillard2024contractions} and holes.
	
	It is precisely at this point that new connected structures appear: components of a co-3-plex may be induced paths or holes, in contrast to the isolated vertices and edges arising for co-2-plexes.
	Asking whether extended formulations can still be obtained from perfect graphs is a natural and fundamental question.
	This work studies the co-3-plex problem on chordal graphs as the smallest graph class where these structural changes can be analyzed while preserving desirable decomposition properties.
	
	\section*{Contributions}       
	
	We build on the co-2-plex extended formulation of~\cite{dupontbouillard2025extendedformulationsmaximumweighted}. Their key insight is based on the observation that the subgraph induced by a co-2-plex is a set of isolated edges and vertices. This allows them to derive an extended formulation of the co-2-plex polytope from the stable set polytope of an auxiliary graph and characterize the co-2-plex polytope of contraction perfect graphs. 
	
	We extend this methodology to co-3-plexes. We propose an extended linear programming formulation for the co-3-plex polytope of chordal graphs with an exponential set of columns. However, we show that the pricing subproblem reduces to finding a maximum vertex and edge weighted induced path in a chordal graph, which we know to be solvable in polynomial time~\cite{dupontbouillard2025extendedformulationsinducedtree}.  This allows us to derive a column generation algorithm to solve the maximum co-3-plex problem. The hardness of this generalization tends to show the limit of such a framework for finding co-$k$-plexes; therefore, we end by discussing the limits of this approach.
	
	\section*{Exponential sized formulation for the co-3-plex polytope of chordal graphs}  \label{sec:co3plex}    
	
	It is important to note that in a chordal graph, the subgraph induced by a co-3-plex is a set of induced paths and triangles.       Let $\mathcal{T}(G)$ be the set of triangles, and $\mathcal{I}(G)$ be the set of induced paths of $G$ with at least one edge. 
	
	We denote by $\mathcal{A}(G)$ the graph having as vertex set $ V \cup \mathcal{T}(G) \cup \mathcal{I}(G)$. and where two vertices are adjacent if their union induces a connected subgraph of $G$ (where vertices of $\mathcal{A}(G)$ associated with elements $v$ of $V$ are considered singletons $\{v\}$). Each node of $\mathcal{A}(G))$ represents a connected component which may belong to a co-3-plex. Moreover, two such components may belong to the same co-3-plex if their union yields a disconnected graph.
	
	\medskip 
	
	Note that $\mathcal{A}(G)$ can be built from $G$ by adding true twins and contracting edges: for each non-singleton $L \in \mathcal{T}(G) \cup \mathcal{I}(G)$ add a twin to each vertex $v \in L$ and contract the added set of twins to obtain vertex $L$ of $\mathcal{A}(G)$. Adding true twins and contracting them is a way to "compress" each subgraph into a single additional node. Figure~\ref{fig:construct} provides an example of this construction on a simple graph.    
	
	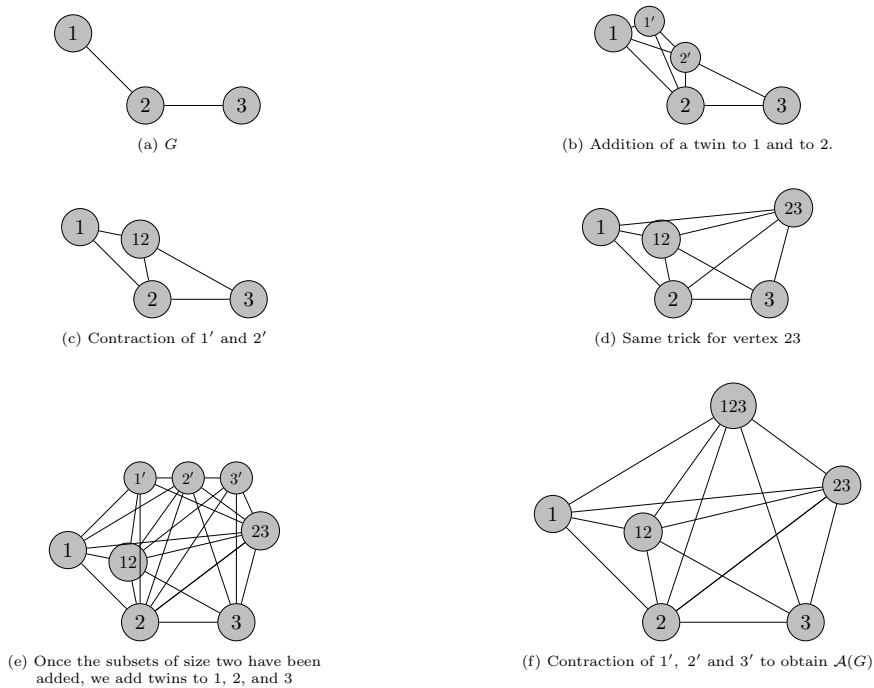
\begin{figure} 
		
		\scalebox{0.8}{\begin{subfigure}[t]{0.5\textwidth}    \centering    
			\begin{tikzpicture}[state/.style={circle, draw, minimum size=0.5 cm},scale = 0.4]          \node[state,draw,circle,fill=gray!50,scale = 1] (1)at(0,3) {$1$};     \node[state,draw,circle,fill=gray!50,scale = 1] (2)at(3,0) {$2$};     \node[state,draw,circle,fill=gray!50,scale = 1] (4)at(7,0) {$3$};          \draw (1) -- (2) ;     \draw (2) -- (4) ;        
			\end{tikzpicture}        \caption{$G$}  
		\end{subfigure}}   \hfill  
		\scalebox{0.8}{\begin{subfigure}[t]{0.5\textwidth}    \centering    
			
			\begin{tikzpicture}[state/.style={circle, draw, minimum size=0.5cm},scale = 0.4]          \node[state,draw,circle,fill=gray!50,scale = 1] (1)at(0,3) {$1$};     \node[state,draw,circle,fill=gray!50,scale = 0.7] (5)at(1.5,3.5) {$1'$};     \node[state,draw,circle,fill=gray!50,scale = 0.7] (6)at(3,2) {$2'$};     \node[state,draw,circle,fill=gray!50,scale = 1] (2)at(3,0) {$2$};     \node[state,draw,circle,fill=gray!50,scale = 1] (4)at(7,0) {$3$};          \draw (1) -- (2) ;     \draw (2) -- (4) ;     \draw(1) -- (5);     \draw (2) -- (5);     \draw (5) -- (6);     \draw (6) -- (2);     \draw (6) -- (4);     \draw (6) -- (1);          \end{tikzpicture}    
		
			\caption{Addition of a twin to $1$ and to $2$.}   \end{subfigure} }
		
		\bigskip   
		\scalebox{0.8}{
		\begin{subfigure}[t]{0.5\textwidth}    \centering   
			
			\begin{tikzpicture}[state/.style={circle, draw, minimum size=0.5cm},scale = 0.4]     \node[state,draw,circle,fill=gray!50,scale = 1] (1)at(0,3) {$1$};     \node[state,draw,circle,fill=gray!50,scale = 0.85] (6)at(2.5,2.5) {$12$};     \node[state,draw,circle,fill=gray!50,scale = 1] (2)at(3,0) {$2$};     \node[state,draw,circle,fill=gray!50,scale = 1] (4)at(7,0) {$3$};          \draw (1) -- (2) ;     \draw (2) -- (4) ;     \draw (6) -- (2);     \draw (6) -- (4);     \draw (6) -- (1);    \end{tikzpicture}    \caption{Contraction of $1’$ and $2’$}  
			
		\end{subfigure}}   \hfill  
		\hfill
		\scalebox{0.8}{\begin{subfigure}[t]{0.5\textwidth}    \centering   
			
			\begin{tikzpicture}[state/.style={circle, draw, minimum size=0.5cm},scale = 0.4]    
				\node[state,draw,circle,fill=gray!50,scale = 1] (1)at(0,3) {$1$};    
				\node[state,draw,circle,fill=gray!50,scale = 0.85] (6)at(2.5,2.5) {$12$};    
				\node[state,draw,circle,fill=gray!50,scale = 1] (2)at(3,0) {$2$};     
				\node[state,draw,circle,fill=gray!50,scale = 1] (4)at(7,0) {$3$}; 
				\node[state,draw,circle,fill=gray!50,scale = 0.85] (7)at(8,3.8) {$23$};             
				\draw (1) -- (2) ;     
				\draw (2) -- (4) ;     
				\draw (6) -- (2);     
				\draw (6) -- (4);     
				\draw (6) -- (1); 
				\draw (4) -- (7);     
				\draw (2) -- (7);     
				\draw (6) -- (7);              
				\draw (1) -- (7);     
				
			\end{tikzpicture}    
			\caption{Same trick for vertex $23$}  
			
		\end{subfigure}  }

		\bigskip
		\scalebox{0.8}{
		\begin{subfigure}[t]{0.5\textwidth}   
			
			\centering   
			
			\begin{tikzpicture}[state/.style={circle, draw, minimum size=0.5cm},scale = 0.4]     
				\node[state,draw,circle,fill=gray!50,scale = 1] (1)at(0,3) {$1$};     
				\node[state,draw,circle,fill=gray!50,scale = 1] (2)at(3,0) {$2$};     
				\node[state,draw,circle,fill=gray!50,scale = 1] (4)at(7,0) {$3$};     
				\node[state,draw,circle,fill=gray!50,scale = 0.85] (5)at(2.5,2.5) {$12$};     
				\node[state,draw,circle,fill=gray!50,scale = 0.85] (7)at(8,3.8) {$23$};     
				\node[state,draw,circle,fill=gray!50,scale = 0.75] (9)at(3,6) {$1’$};     
				\node[state,draw,circle,fill=gray!50,scale = 0.75] (10)at(5,6) {$2’$};     \node[state,draw,circle,fill=gray!50,scale = 0.75] (11)at(7,6) {$3’$};     
				\draw (1) -- (2) ;     
				\draw (2)-- (4) ; 
				\draw (10) -- (1);     
				\draw (5) -- (7);     
				\draw (1) --(5);     
				\draw (2) --(5);     
				\draw (4) -- (7);     
				\draw (2) -- (7);     
				\draw (2) -- (7);              
				\draw (1) -- (7);     
				\draw (4) -- (5);          
				\draw (1) -- (9);     
				\draw (2) -- (9);     
				\draw (5) -- (9);     
				\draw (7) -- (9);          
				\draw (2) -- (10);     
				\draw (4) -- (10);     
				\draw (5) -- (10);     
				\draw (7) -- (10);          
				\draw (2) -- (11);     
				\draw (4) -- (11);    
				\draw (5) -- (11);     
				\draw (7) -- (11);          
				\draw (9) -- (10);     
				\draw (10) -- (11);        
				
			\end{tikzpicture}
			
			\caption{Once the subsets of size two have been added, we add twins to $1$, $2$, and $3$ }      
			
		\end{subfigure}  }   
		\hfill
		\scalebox{0.8}{\begin{subfigure}[t]{0.5\textwidth}    \centering    \begin{tikzpicture}[state/.style={circle, draw, minimum size=0.5cm},scale = 0.6]     \node[state,draw,circle,fill=gray!50,scale = 1] (1)at(0,3) {$1$};     \node[state,draw,circle,fill=gray!50,scale = 1] (2)at(3,0) {$2$};     \node[state,draw,circle,fill=gray!50,scale = 1] (4)at(7,0) {$3$};     \node[state,draw,circle,fill=gray!50,scale = 0.85] (5)at(2.5,2.5) {$12$};     \node[state,draw,circle,fill=gray!50,scale = 0.85] (7)at(8,3.8) {$23$};     \node[state,draw,circle,fill=gray!50,scale = 0.85] (10)at(5,6) {$123$};     \draw (1) -- (2) ;     \draw (2) -- (4) ;          \draw (5) -- (7);          \draw (1) --(5);     \draw (2) --(5);     \draw (4) -- (7);     \draw (2) -- (7);     \draw (2) -- (7);               \draw (1) -- (7);     \draw (4) -- (5);               \draw (1) -- (10);     \draw (2) -- (10);     \draw (4) -- (10);     \draw (5) -- (10);     \draw (7) -- (10);               \color{black}         \end{tikzpicture}    \caption{Contraction of $1’, \ 2’$ and $3’$ to obtain $\mathcal{A}(G)$}   \end{subfigure} } \caption{Construction of $\mathcal{A}(G)$ by adding twins and contracting edges.}   \label{fig:construct}     \end{figure}  
	
	\begin{obs}\label{obs:bijection}   There is a one-to-one correspondence between the  co-3-plexes of a chordal graph $G$ and the stable sets of $\mathcal{A}(G)$.  \end{obs} 
	
	\begin{proof}   Let $S$ be a co-3-plex of $G$, and $(L_1,\dots,L_\ell)$ be a partition of $S$ into the connected components of $G[S]$.   As the connected components of a co-3-plex are either induced paths or triangles, we have that $L_1,\dots,L_\ell$ belong to the vertex set of $\mathcal{A}(G)$. No two elements  $L_i, \ L_j$ of $K_\mathcal{A}$ are adjacent, otherwise $G[L_i \cup L_j]$ would be connected.      Given a stable set $S_{\mathcal{A}} = \{ L_1,\dots, L_\ell\}$ of $\mathcal{A}(G)$, each vertex $L_j$ induces a connected component of $G$ Moreover, $G[\bigcup_{i=1}^\ell L_i]$ has exactly $\ell$ connected components, all of degree at most 2, by construction. This implies that $\bigcup_{i=1}^\ell L_i$ induces a co-3-plex. 
	\end{proof}
	
	As adding true twins and contracting edges preserves chordality~\cite{dupontbouillard2024contractions}, and the set of induced paths of size 1 induces a subgraph of $\mathcal{A}(G)$ isomorphic to $G$, we obtain the following.   
	
	\begin{corollary}\label{the:characchordal}   A graph $G$ is chordal if and only if $\mathcal{A}(G)$ is.  \end{corollary}      
	Another important remark is that the set of maximal cliques of a chordal graph $G$ and the one of $\mathcal{A}(G)$ are also in one-to-one correspondence.    
	
	\begin{lemma} \label{lem:cliquechordal}   Given a chordal graph $G = (V,E)$, the maximal cliques of $\mathcal{A}(G)$ are of the following form:   
		\begin{equation}    K  \cup \{L \in V(\mathcal{A}(G)) \setminus V : L \cap K \neq \emptyset \}  \quad \forall K \in \mathcal{K}(G). \label{eq:cliqueAg}   \end{equation}
	\end{lemma} 
	\begin{proof}   The proof is done by showing the stronger result stating that the maximal cliques of a graph $G'$, obtained by adding twins to a graph $G = (V,E)$ and contracting them, are of the form $K  \cup \{L \in V(G') \setminus V : L \cap K \neq \emptyset \} \quad  \forall K \in \mathcal{K}(G)$.
		This  implies that equality~\eqref{eq:cliqueAg} holds at each step of  $\mathcal{A}(G)$'s construction.

	By contradiction, let $G$ be a chordal graph and suppose that adding a minimal set of twins to $G$ and contracting them yields a graph $G'$ with a new vertex $L$ and a maximal clique $K_{G'}$ containing $L$ such that $K =V(G') \cap V$ has an empty intersection with $L$. Minimality implies that there exists a vertex $w_1$ of $K$ adjacent to exactly one vertex $v$ of $L$. Let $P$ be a path of $L$ starting in $v$ whose other extremity is adjacent to an element $w_2$ of $K$ and whose interior is non-adjacent to any element of $K$. By construction, $P \cup \{w_1,w_2\}  $ induces a hole of size at least 4 of $G'$, a contradiction to Corollary~\ref{the:characchordal}.   
	\end{proof}  
	
	Now, we give the exponential dimensional space formulation of the co-3-plex polytope of chordal graphs. Let $\mathcal{T}_K^p$ denote the set of triangles and $\mathcal{I}_K^p$ the set of induced paths containing at least $p$ elements of $K$. We introduce variable $p$ (resp. $t$) associated with the induced path (resp. triangles) of $G$. We encode a co-3-plex $S$ of $G$ by the vector of $x, \ p,$ and $t$ variables so that $x_v = 1 \ \forall v \in V $ and $p_P = 1$ (resp. $t_T = 1$) if $P$ (resp. $T$) induces a connected component of $G[S]$.
	   
	\begin{theorem}  
		A graph $G$ is chordal if and only if the following polytope is integer~:  
		\begin{align}   
			x(K)- t (\mathcal{T}_K^2) - 2 t(\mathcal{T}_K^3) - p(\mathcal{I}_K^2)  \le 1 && \forall K \in \mathcal{K}(G) \label{iq:enoncecliqueTheCo3Plex}\\
			t(\mathcal{T}_{v}^1) + p(\mathcal{I}_{v}^1) \le x_v && \forall v \in V \label{iq:enoncestarTheCo3Plex}\\ 
			p,t \ge 0 \label{iq:enoncetrivialTheCo3plex} 
		\end{align}     
	\end{theorem} 
	\begin{proof}   $(\Rightarrow)$    Suppose that $G$ is chordal, then $\mathcal{A}(G)$ also is by Corollary~\eqref{the:characchordal}.  Let us consider the following set of variables: $z_u \ \forall u \in V$, $t_T \ \forall T \in \mathcal{T}(G)$, \ $p_I, \ \forall I \in \mathcal{I}(G)$ associated with the vertices of $\mathcal{A}(G)$. Using Lemma~\ref{lem:cliquechordal} and Theorem~\ref{the:WPGT} we get that the stable set polytope of $\mathcal{A}(G)$ is described by the following system :   
		\begin{align}    z(K) + t(\mathcal{T}_K^1) + p(\mathcal{I}_K^1) \le 1 && \forall K \in \mathcal{K}(G) \label{iq:cliqueTheCo3plex}\\    z,t,p \ge 0  \label{iq:trivialTheCo3plex}     \end{align}         The bijection pointed out in Observation~\ref{obs:bijection} maps a stable set $S_{\mathcal{A}}$ of $\mathcal{A}(G)$ to a co-3-plex $S$ of $G$ by the following : $\bigcup_{L \in S_{\mathcal{A}}} L = S $. Following this idea, we can express variables $x_u, \ \forall u \in V$ as a combination of $z, p$ and $t$.  The following relation expresses the fact that a co-3-plex of $G$ contains a vertex $v$ if and only if its associated stable set of $\mathcal{A}(G)$ contains exactly one element containing $v$.   
		\begin{equation}    x_u = z_u + t(\mathcal{T}_{u}^1) + p(\mathcal{I}_{u}^1) \ \forall u \in V.    \label{equation:rel}    \end{equation} 
		
		We use equality~\eqref{equation:rel}, we get rid of $z$ variables in \eqref{iq:cliqueTheCo3plex}--\eqref{iq:trivialTheCo3plex} to obtain \eqref{iq:enoncecliqueTheCo3Plex}--\eqref{iq:enoncetrivialTheCo3plex} as follows:  
		\begin{itemize}    
			\item inequalities~\eqref{iq:enoncecliqueTheCo3Plex} are obtained from \eqref{iq:cliqueTheCo3plex} by noticing that each variable $p$ or $t$ associated with an element, say $P$ such that $i = |P \cap K| \ge 1$  ends up with a coefficient $-i+1$ (each $z_v$ for $v \in P\cap K$ will be replaced by equality~\eqref{equation:rel}) 
			\item inequalities~\eqref{iq:enoncestarTheCo3Plex} appear by combining \eqref{equation:rel} and the nonnegativity constraints on $z$.  
		\end{itemize}  
		
		$(\Leftarrow)$   By contradiction, suppose that $G$ is not chordal. 
		Suppose that $G$ contains a hole induced by $H$ of length at least 5, if it is odd then $H$ induces an odd hole of $\mathcal{A}(G)$ and otherwise given any two adjacent vertices $u, \ v$ of $H$, we get that $(H \cup \{uv \}) \setminus \{u,v\}$  induces an odd hole of $\mathcal{A}(G)$. 
		In both cases, $\mathcal{A}(G)$ is not perfect by the strong perfect graph theorem, and backtracking the previous proof yields that \eqref{iq:enoncecliqueTheCo3Plex}--\eqref{iq:enoncetrivialTheCo3plex} is not integer. 
		If $G$ contains a hole induced by $(u,v,w,t)$, we prove that the point $(x^*,y^*)$ satisfying $x_u^* = x_v^* = x_{w}^* = x_t^* = p_{wt}^* = \frac{1}{2}$ and whose other components are equal to 0 is a fractional extreme point of \eqref{iq:enoncecliqueTheCo3Plex}--\eqref{iq:enoncetrivialTheCo3plex}. 
		First, we prove that $(x^*,y^*)$ is a point of \eqref{iq:enoncecliqueTheCo3Plex}--\eqref{iq:enoncetrivialTheCo3plex}: notice that no clique of $G$ contains $u, \ v$ and an extremity of $wt$, implying that \eqref{iq:enoncecliqueTheCo3Plex} are satisfied by $(x^*,t^*,p^*)$ and all other inequalities are trivially satisfied. 
		Finally, we exhibit $|V| + |\mathcal{T}(G)| + |\mathcal{I}(G)|$ linearly independent tight constraints for $(x^*,t^*,p^*)$ :  
		
		\begin{itemize}    
			\item inequalities~\eqref{iq:enoncetrivialTheCo3plex} for variables associated with elements of $V \cup \mathcal{I}(G) \cup \mathcal{T}(G) \setminus \{wt \}$    
			\item inequalities~\eqref{iq:enoncecliqueTheCo3Plex} for any three maximal cliques of $G$ respectively containing $ \{ u,v \}$, $\{v,w\}$, and $\{u,t\}$    
			\item inequalities~\eqref{iq:enoncestarTheCo3Plex} for $w$ and $t$.   
		\end{itemize}   
	\end{proof}     
	
	As the orthogonal projection of an integer polytope yields another integer polytope, we get the following.
	\begin{corollary}
		The co-3-plex polytope of a chordal graph is equal to the projection of the linear system~\eqref{iq:enoncecliqueTheCo3Plex}--\eqref{iq:enoncetrivialTheCo3plex} onto the $x$ variable space.
	\end{corollary}
	
	As the number of maximal cliques of chordal graphs is linear, Formulation~\eqref{iq:enoncecliqueTheCo3Plex}--\eqref{iq:enoncetrivialTheCo3plex} has a polynomial number of constraints. The number of triangles is also polynomial; however, the number of induced paths is exponential. Therefore, we propose a column generation algorithm in which the pricing subproblem is to find an induced path of $G$ whose variable has the maximum reduced cost.  Note that a variable $p_P$ is involved in a constraint~\eqref{iq:enoncecliqueTheCo3Plex} only if it intersects the clique $K$ twice, meaning that it contains an edge of $K$. Moreover, at most two vertices of an induced path can belong to a clique. Given input weights on the vertices $w\in \mathds{Q}^V$, dual values $\lambda_K \ \forall K \in \mathcal{K}(G)$ and $\mu_v \ \forall v \in V$ associated with the two non-trivial sets of constraints, the pricing subproblem is to find a maximum induced path $P$ whose weight is equal to :   
	
	\begin{equation}   \sum_{v \in P} w_v -\mu_v + \sum_{uv \in E(P)} \sum_{K \in \mathcal{K} : \{u,v\} \subseteq K} \lambda_K   \end{equation}   
	
	The pricing subproblem is then to find a maximum induced path with weights on vertices and edges in a chordal graph, which is doable in polynomial time as shown by the author of~\cite{dupontbouillard2025extendedformulationsinducedtree}. Therefore, we obtain the following, owing to the equivalence between optimization and separation~\cite{separationOptimization}:    \begin{corollary}   The weighted maximum co-3-plex problem is solvable in polynomial time on chordal graphs.  
	\end{corollary}
	
	\medskip
	
	In the following, we give arguments tending to discard potential generalizations of this framework, reducing the maximum co-3-plex problem to some maximum stable set problem in an auxiliary perfect graph.  
	First, note that a consequence of Theorem~2.2 in~\cite{dupontbouillard2024contractions} is that $\mathcal{A}(G)$ is perfect if and only if $G$ is contraction perfect, which implies that contraction perfect graphs are the larger class of graphs on which this type of approach could be used.  
	
	\medskip
	
	The holes of contraction perfect graphs have size at most 4~\cite{dupontbouillard2024contractions}. 
	One could consider a generalization of the system~\eqref{iq:cliqueTheCo3plex}--\eqref{iq:trivialTheCo3plex} with variables associated with triangles and holes of size 4. 
	It requires solving the maximum vertex and edge weighted induced path problem in contraction perfect graphs, which is unknown to be polynomially solvable. When the edges are unweighted, the maximum weight induced path is solvable in polynomial time on $t$-chordal graphs~\cite{ISHIZEKI20083057}, which form a superclass of contraction perfect graphs for $k=4$. Unfortunately, it is unknown whether this algorithm can be extended to vertex and edge weighted graphs. 
	
	Another issue with such an approach is that the maximal cliques of $\mathcal{A}(G)$ for some contraction-perfect graph $G$ are not in one-to-one correspondence with the maximal cliques of $G$. 
	Their structure is related to holes of size 4; an example is given for some hole $(u,v,w,t)$ for which $\{u,v,wt\}$ induces a triangle whose completions into maximal cliques  of $\mathcal{A}(G)$ never contain $\{w,t\}$.
	However, besides being associated with holes of size 4, their number is no longer polynomial in contraction perfect graphs~\cite{dupontbouillard2025extendedformulationsmaximumweighted}, making it mandatory to develop a price and cut algorithm.     
	
	\medskip  
	
	Another generalization idea would be for general $k$. Such an extension corresponds to finding a maximum stable set in a graph whose vertices are subsets of $V(G)$, inducing connected subgraphs of maximum degree $k-1$, where two vertices are adjacent if they induce a connected component of $G$. The generalizations of Theorem~\ref{the:characchordal} and Lemma~\ref{lem:cliquechordal} would still hold by the adding twin/contracting edges construction. 
	Such sets of vertices could intersect the maximal cliques of $G$ more than twice, implying coefficients as low as $(-k+1)$ in the corresponding inequalities~\eqref{iq:enoncecliqueTheCo3Plex}.
	For general k, the reduced cost captures the number of times S hits a clique K beyond the first vertex.
	The pricing subproblem would then be to find a subset of vertices inducing a connected subgraph maximizing a function $f(S)$ depending on vertex weight $\lambda_v$ (dual values of inequalities~\eqref{iq:enoncestarTheCo3Plex}) and clique weight $\mu_K$ (dual values of the corresponding inequalities~\eqref{iq:trivialTheCo3plex}), and is as follows: $f(S) = w(S) - \lambda(S) + \sum_{K \in \mathcal{K}(G)} \mu_K \max(0, |S \cap K| -1) $.      
	
	\medskip     
	
	The description we provide in this work for the co-3-plex polytope of chordal graphs is implicit, and it could be interesting to project it onto smaller spaces to obtain, in particular explicit description of its facets. In~\cite{dupontbouillard2025extendedformulationsmaximumweighted}, a  compact extended formulation for the co-2-plex polytope of chordal graphs is given, which corresponds to Formulation~\eqref{iq:cliqueTheCo3plex}--\eqref{iq:trivialTheCo3plex}, where all triangles and induced paths of size larger than 3 have their variables set to 0. The author projected it onto the natural variable space for trees.  Obtaining the natural variable space of the co-3-plex polytope of chordal graphs is a harder task than providing it for the co-2-plex polytope by the correspondence between both extended formulations. The last question around this work would be to provide the extension complexity of the co-3-plex polytope of chordal graphs, there is still hope for a compact formulation to exist.       
	
	\section*{Conclusion}    We give a new application of the maximum vertex and edge weighted induced path problem to solve the maximum co-3-plex problem of a contraction perfect graph via a column generation algorithm. We apply this framework to the particular case of chordal graphs for which the pricing subproblem is polynomially solvable, as it reduces to finding a maximum weight induced path in a chordal graph, which yields a polynomial algorithm for the problem. Finally, we discuss the limit of such a generalization by raising the principal obstacles.

	\section*{Acknoledgment}    The author thanks the ANR project ANR-24-CE23-1621 EVARISTE for financial support.

	\bibliographystyle{plain}
	\bibliography{bibtex}
	
\end{document}